\documentclass[11pt]{article}

\newif\iffull 
\fulltrue

\usepackage{url}
\usepackage{algorithmic}
\usepackage{algorithm}
\usepackage[]{graphicx}
\usepackage[]{times}

\usepackage{fullpage}
\usepackage{amsmath}
\usepackage{amssymb}

\newcommand{\mtx}[1]{\mathbf{#1}}
\newcommand{\meas}{{\mathbf \mu}}
\newcommand{\noise}{{\mathbf \nu}}

\newcommand{\signal}{{\mathbf x}}

\newcommand{\signaly}{{\mathbf y}}
\newcommand{\signalz}{{\mathbf z}}
\newcommand{\nerr}[1]{\left\| #1 \right\|}
\newcommand{\supp}[1]{\operatorname{supp}(#1)}
\newcommand{\mysplit}[1]{\operatorname{split}(#1)}
\newcommand{\halmos}{\hspace*{\fill}\rule{1ex}{1.4ex}}

\newcommand{\tightpgh}[1]{\vspace*{-1\baselineskip}\paragraph{#1}}

\newenvironment{proof}{\noindent {\bf Proof.}}{\halmos}

\newtheorem{thm}{Theorem}

\newtheorem{lemma}[thm]{Lemma}

\newtheorem{definition}[thm]{Definition}
\newtheorem{invariant}[thm]{Invariant}

\title{Sublinear Time, Measurement-Optimal, Sparse Recovery For All}

\author{Ely Porat\thanks{Bar Ilan University.  Email: {\tt
      porately@gmail.com}}
 ~and Martin J. Strauss\thanks{University of
    Michigan.  Supported in part by NSF CCF 0743372 and DARPA/ONR
    N66001-08-1-2065.  Email: {\tt martinjs@eecs.umich.edu}}}

\begin{document}

\maketitle

\begin{abstract}
An {\em approximate sparse recovery} system in $\ell_1$ norm makes a
small number of
measurements of a noisy vector with at most $k$ large entries and
recovers those {\em heavy hitters} approximately.
Formally, it consists of parameters
$N,k,\epsilon$, an $m$-by-$N$ {\em measurement matrix}, $\mtx{\Phi}$,
and a
{\em decoding algorithm}, $\mathcal{D}$.  Given a vector, $\signal$,
where $\signal_k$ denotes the optimal $k$-term approximation to
$\signal$, the system approximates $\signal$ by $\widehat
\signal=\mathcal{D}(\mtx{\Phi} \signal)$, which must satisfy
\[\nerr{\widehat \signal - \signal}_1
\le (1+\epsilon)\nerr{\signal - \signal_k}_1.\]
Among the goals in
designing such systems are minimizing the number $m$ of measurements
and the runtime of the decoding algorithm, $\mathcal{D}$.  We consider
the ``forall'' model, in which a single matrix $\mtx{\Phi}$, possibly
``constructed'' non-explicitly using the probabilistic method, is used
for all signals $\signal$.

Many previous papers have provided algorithms for this problem.  But
all such algorithms that use the optimal number $m=O(k\log(N/k))$ of
measurements require 
superlinear time $\Omega(N\log(N/k))$.  In this paper, we give the
first algorithm for this problem that uses the optimum number of
measurements (up to
constant factors) and runs in sublinear time $o(N)$ when $k$ is
sufficiently less than $N$.
Specifically, for
any positive integer $\ell$, our approach uses time
$O(\ell^{5}\epsilon^{-3}k(N/k)^{1/\ell})$ and uses
$m=O(\ell^{8}\epsilon^{-3}k\log(N/k))$ measurements, with access to a
data structure requiring space and preprocessing time
$O(\ell Nk^{0.2}/\epsilon)$.
\end{abstract}

\thispagestyle{empty}
\setcounter{page}{0}

\newpage

\section{Introduction}

\subsection{Description of Problem}

Variations of the
Sparse Recovery problem are well-studied in recent literature.  A
vector (or signal) $\signal$ is first measured, by the matrix-vector
product $\meas=\mtx{\Phi}\signal$, then, at a later time, a
decoding algorithm $\mathcal{D}$ approximates $\signal$ from
$\meas$.  The approximation is non-vacuously useful if $\signal$
is dominated by
a small number of large magnitude entries, called ``heavy hitters.''
Applications arise in signal and image processing and database, with
further application to telecommunications and
medicine~\cite{singl_pixel,lustig07:_spars_mri}.  Several
workshops~\cite{carin09:_duke_works,spars09:_spars_works}
have been devoted to this topic.  See more at~\cite{rice}.

In this paper, we focus on the following variation.
If $N$ is the
length of the signal, $k$ is a sparsity parameter, and $\epsilon$ is a
fidelity parameter, we want
$\nerr{\widehat \signal - \signal}_1
\le (1+\epsilon)\nerr{\signal_k - \signal}_1$,
where $\signal_k$ is the best possible $k$-term representation for
$\signal$.
Among the goals in
designing such systems are minimizing the number $m$ of measurements
and the runtime of the decoding algorithm, $\mathcal{D}$.  We consider
the ``forall'' model, in which a single matrix $\mtx{\Phi}$, possibly
``constructed'' non-explicitly using the probabilistic method in
polynomial time or explicitly in exponential time, is used
for all signals $\signal$.

\subsection{Advantages over Previous Work}

\vspace*{1\baselineskip}
\tightpgh{Previous measurement-optimal algorithms are slow.}
Many previous papers have provided algorithms for this problem.  But
all such algorithms that use the constant-factor-optimal number
$O(k\log(N/k))$ of measurements require 
superlinear time $\Omega(N\log(N/k))$.  In this paper, we give the
first algorithm that, for any positive integer $\ell$, uses
$\ell^{O(1)}\epsilon^{-3}k\log(N/k)$
measurements ({\em i.e.}, $O(k\log(N/k))$ measurements for constant
$\ell$ and $\epsilon$) and
run in time $\ell^{O(1)}\epsilon^{-3}k(N/k)^{1/\ell}$.  For example,
with $\ell=2$ and $\epsilon=\Omega(1)$,
the runtime improves from $N$ to $\sqrt{kN}$.  In some
applications,
sparse recovery is the runtime bottleneck and our contribution can
make some other $\Omega(N)$ computation become the new bottleneck.

The sublinear runtime of our algorithm is important not because
traditional algorithms are too slow, but because the
measurement-optimal algorithms that replaced them are too slow.
Consider an application in which $k\ll N$.  A traditional approach
makes exactly $N$ direct measurements or
(in some cases) requires little more than taking a single Fast
Fourier Transform of length $N$.  Optimized code for FFTs is so fast
that one cannot plausibly claim to lower the runtime, say from $N\log
N$ to $\sqrt{N}$, by a complicated algorithm with heavy overhead.  But, in some
cases (see below), taking more measurements than necessary is a
significant liability.  Several papers in the literature (see
Table~\ref{table:previous}) improve the number of measurements from
$N$ to $k\log(N/k)$, but only with {\em significant increase} in the
runtime, say, from computing a Fourier Transform to solving a
linear program or, more recently, performing a combinatorial algorithm
on expander graphs, of a flavor similar to our approach below.  When
$k$ is small compared with $N$, we
hope that (i) the number $O(k\log(N/k))$ of measurements made by our
algorithm is significantly less than $N$ in practice, and (ii) the
sublinear runtime of our algorithm is significantly faster than that of
other
{\em measurement-optimal} algorithms, all of which use time
$\Omega(N\log(N/k))$, and many of which have significant overhead.  We
do not expect that our ``sublinear time'' algorithm will compete on
time with naive time $O(N)$ algorithms or with a single FFT, except
for in unusual circumstances and/or values of $k$ and $N$.

These questions have been actively studied by
several communities.
See Table~\ref{table:previous}, which is based in part on a table
in~\cite{RI08}.

\tightpgh{Trading runtime for fewer measurements in sublinear-time
  algorithms.}
Previous sublinear-time algorithms for this problem have used too many
measurements by logarithmic factors, which we now argue is
inappropriate in certain situations.  In a traditional approximation
algorithm, there is an objective function to be minimized, and relatively
small improvements in the approximation ratio for the objective
function---from $O(\log n)$ to constant-factor to $(1+o(1))$---are
considered well worth a polynomial blowup in computation time.  In the
sparse recovery problem, there are two objective functions---the
approximation ratio by which the error
$\nerr{\widehat\signal-\signal}_1$ exceeds the optimal, and the number
of measurements.  In case of medical imaging, the number of
measurements is proportional to the duration during which a patient
must lie motionless; a measurement blowup
factor of ``1000 times log of something'' is
unacceptable.
In this paper, we reduce the blowup in number of
measurements to a small integer constant factor.

Previous sublinear time algorithms had runtime polynomial in
$k\log(N)$, often linear in $k$.  By contrast, our algorithm gives
runtime $\epsilon^{-3}\sqrt{kN}$ or, more generally, gives runtime
$\ell^{O(1)}\epsilon^{-3}k(N/k)^{1/\ell}$ using
$\ell^{O(1)}\epsilon^{-3}k\log(N/k)$ measurements, which
remains slightly suboptimal.  But, first, a blowup in
runtime from, say, $k\log^2 N$ to $k(N/k)^{1/4}$ is appropriate to
reduce the approximation ratio from logarithmic to small constant in a
critical objective like number of measurements in certain
applications.  Second, the blowup is not that big in other
applications, where, say, $(N/k)^{1/4}$ is not much bigger than
$\log^2(N)$.  Alternatively, a parametric sweet
spot for our
algorithm occurs 
around $k=N^{1/4}$.  Putting $\ell=3$, we get
runtime $k(N/k)^{1/3}=\sqrt{N}=k^2$.  This is about the time to
multiply a vector by a dense matrix of smallest useful size,
which is a tiny component in some (early) algorithms in the
literature, superlinear or sublinear.

\tightpgh{Constant factor gap in number of measurements.}
The best previous {\em superlinear} algorithms~\cite{Roman-CISS} use a number of
measurements that is suboptimal by a small constant factor versus the
best known lower bounds.  Thus, for
sublinear-time algorithms, a small constant-factor gap, rather than an
approximation scheme, is currently an
appropriate goal.

\tightpgh{All signals or Each signal?}
The results of this paper are in the ``forall'' model.
Recently, a sublinear-time, constant-factor-optimal measurement
algorithm was given~\cite{GLPS10} in an incomparable setup.  In
particular, its guarantees were for the weaker ``foreach'' model, in
which a random measurement matrix works with {\em each} signal, but no
single matrix works simultaneously on {\em all} signals.\footnote{In
  the forall model, the guarantee is that a matrix $\mtx{\Phi}$
  generated
  according to a specified distribution succeeds on {\em all} signals
  in a class $\mathcal{C}$.  In the foreach model, there is a
  distribution on matrices, such that for {\em each} signal 
  $\signal$ in a class $\mathcal{C'}$ bigger than $\mathcal{C}$, a
  matrix $\mtx{\Phi}$ chosen according to the prescribed distribution
  succeeds on $\signal$.  The difference in models is captured
  in the order of quantifiers, which can be anthropomorphized into
  the powers of a challenger and adversary.}
The stronger
forall model is more appropriate in certain applications, where, for
example, there is a sequence $\signal^{(1)}, \signal^{(2)}$ of signals
to be measured by the same measurement matrix, and $\signal^{(2)}$
depends, in some subtle way, on the result of recovering
$\signal^{(1)}$.  (For example, an adversary may construct
$\signal^{(2)}$ after observing an action we take in response to
recovering $\signal^{(1)}$.)  In the forall model, there is no issue.
In the foreach model, however, it is important that an adversary pick
the signal without knowing the outcome
$\mtx{\Phi}$.  If the adverary knows something about the outcome
$\mtx{\Phi}$---such as observing our reaction to recovering $\signal^{(1)}$
from $\mtx{\Phi}\signal^{(1)}$---the adversary may be able to
construct an $\signal^{(2)}$ in the null space of $\mtx{\Phi}$, which
would break an algorithm in the weaker foreach model.

\begin{table*}
{\centering\footnotesize
\begin{tabular}{|c|c|c|c|c|c|c|}
\hline
Paper           & A/E & No. Measurements  & Column sparsity/ & Decode time  & 	Approx. error & Noise\\
 & & & Update time & & & \\
\hline
\cite{CCF}      & E   & $k \log^c N$      & $\log^c N$       & $N \log^c N$ & $\ell_2 \le C \ell_2$ & \\
\hline
\cite{CM06:Combinatorial-Algorithms}
                & E   & $k \log^c N$      & $\log^c N$       & $k \log^c N$ & $\ell_2 \le C \ell_2$ & \\
\hline
\cite{CM03b}    & E   & $k \log^c N$      & $\log^c N$       & $k \log^c N$ & $\ell_1 \le C \ell_1$ & \\
\hline
\cite{GLPS10}   & E   & $k \log(N/k)$     & $\log^c N$       & $k\log^c N$  & $\ell_2 \le C \ell_2$ & Y \\
\hline
\cite{Don06:Compressed-Sensing,CRT06:Stable-Signal}
                & A   & $k\log( N/k)$     & $k\log(N/k)$     & LP           & $\ell_2 \le (C/\sqrt{k}) \ell_1$ & Y \\
\hline
\cite{GSTV06}   & A   & $k\log^c N$       & $\log^c N$       & $k\log^c N$  & $\ell_1\le (C\log N)\ell_1$   & Y \\
\hline
\cite{GSTV07:HHS}
                & A   & $k\log^c N$       & $k\log^c N$      & $k^2\log^c N$& $\ell_2\le (\epsilon/\sqrt{k})\ell_1$   & Y \\
\hline
\cite{RI08}     & A   & $k\log(N/k)$      & $\log(N/k)$      & $N\log(N/k)$ & $\ell_1\le(1+\epsilon)\ell_1$ & Y\\
\hline
\hline
This paper
                & A   & $\ell^ck\log(N/k)$
                                          & $(\ell\log N)^c $
                                                             & $\ell^c k(N/k)^{1/\ell}$
                                                                            & $\ell_1\le(1+\epsilon)\ell_1$ & Y\\

(any integer $\ell$)
                &     &                   &                  &              &                               & \\
\hline
\hline
Lower bound ``A''
                & A   & $k\log(N/k)$      & $\log(N/k)$      & $k\log(N/k)$ & $\ell_2\le(\epsilon/\sqrt{k})\ell_1$ & Y\\
\hline
\end{tabular}
\caption{\footnotesize Summary of the best previous results and the
  result obtained in this paper.  Some constant factors are omitted
  for clarity.  ``LP'' denotes (at least) the time to do a linear
  program of size at least $N$.  The column ``A/E'' indicates whether
  the algorithm works in the forall (A) model or the foreach (E)
  model.  The column ``noise'' indicates whether the algorithm
  tolerates noisy measurements.  Measurement and decode time dependence on
  $\epsilon$, where applicable, is polynomial.}
\label{table:previous}
}
\end{table*}

\subsection{Overview of Results and Techniques}

First, following previous work~\cite{GLPS10}, we show that it suffices
to recover all but approximately $k/2$ of $k$ heavy hitters at a time.  The
cost for this in measurements is $ck\log N/k$, for some constant $c$.
We then repeat on the remaining $k'=k/2$ heavy hitters, with cost
$ck'\log N/k'\approx\frac12ck\log N/k$, and leaving $k/4$ heavy
hitters.  Continuing this way, the total cost is a geometric
progression with sum $O(k\log N/k)$.  In fact, we will use somewhat
more than $\frac12 ck'\log N/k'$ measurements, {\em e.g.},  $\frac9{10}
ck'\log N/k'$ measurements, to enforce other requirements while still
keeping the number of measurements bounded by a geometric series that
converges to $O(k \log N/k)$.  As in~\cite{GLPS10}, we present a
compound loop invariant satisfied as the number of heavy hitters drops
from $k$ to $k/2$ to $k/4$, etc.

Next, we show how to solve the transformed problem, {\em i.e.}, how to
reduce the number of unrecovered heavy hitters from $k$ to $k/2$,
while not increasing the noise by much.  As in previous results, we
estimate all $N$ coefficients of $\signal$ by hashing the positions
into $O(k)$ buckets, hoping that each heavy hitter ends up dominating
its bucket, so that the bucket aggregrate is a good estimate of the
heavy hitter.  We repeat $O(\log(N/k))$ times, and take a median of
estimates.  Finally, we replace by zero all but the largest $O(k)$
estimates.  If all estimates were independent, then this would give
the result we need, by the Chernoff bound; below we handle the minor
dependence issues.
We get a simple and natural system making $O(k\log N/k)$
measurements but with runtime
somewhat larger than $N$.

Finally, to get a sublinear time algorithm, we replace the above
exhaustive search over a space of size $N$ with constantly-many
searches over
spaces of size approximately $\sqrt{kN}=k(N/k)^{1/2}$.  Still more
generally, replace with $\ell^{O(1)}$ searches over spaces of size
$\ell^{O(1)}k(N/k)^{1/\ell}$, for any positive integer value of the
user-parameter $\ell$.  As a tradeoff, this requires the factor
$\ell^{O(1)}$ times more measurements.  In the case $\ell=2$, we first
hash the original signal's indices into $\sqrt{kN}$ buckets, forming a
new signal $\signal'$, indexed by buckets.  As we show, heavy hitters
in $\signal$ are likely to dominate their buckets, which become heavy
hitters of $\signal'$.   We then find approximately $k$ heavy buckets
exhaustively, searching a space of size $\sqrt{kN}$.  Each bucket
corresponds to approximately $N/\sqrt{kN}=\sqrt{N/k}$ indices in the
original signal, for a total of $k\sqrt{N/k}=\sqrt{kN}$ indices, which
are now searched.  This naturally leads to
runtime $\sqrt{kN}\log(N/k)$, or $k(N/k)^{1/\ell}\log(N/k)$ for
$\ell>2$.  By absorbing $\log(N/k)$ into $\ell^{O(1)}(N/k)^{1/\ell}$,
we get, for general $\ell$ and $\epsilon$,
\begin{thm}
For any positive integer $\ell$, there is a solution to the $\ell_1$
forall sparse recovery problem running in time
$O(\ell^{5}\epsilon^{-3}k(N/k)^{1/\ell})$ and using
$O(\ell^{8}\epsilon^{-3}k\log(N/k))$
measurements,
where $N$ is the length, $k$ is the sparsity, and $\epsilon$ is the
approximation parameter.  The algorithm uses a data structure that
requires space and preprocessing time $O(\ell Nk^{0.2}/\epsilon)$.
\end{thm}

Note:  For thoroughness, we count the factors of $\ell$ and
$1/\epsilon$.  The reader is warned, however, that we are aware of
possible improvements, so the reader may want to abstract $\ell^8/\epsilon^3$ and
$\ell^5/\epsilon^3$ to the simpler expression
$(\ell/\epsilon)^{O(1)}$.  Similarly, the power $0.2$ of $k$ in the
preprocessing costs can be improved but with constant-factor increases
to runtime or number of measurements, and we suspect that expensive
preprocessing can be eliminated altogether.  (The focus of this paper
is just sublinear runtime and constant-factor-optimal number of
measurements, while other aspects of the algorithm are reasonable but
not optimal.)

\subsection{Organization of this paper}

This paper is organized as follows.  Note that we are specifying a
measurement matrix and a decoding algorithm; we refer to the
combination as a {\em system}.  In Section~\ref{sec:defs}, we present
notation and definitions.  In Section~\ref{sec:main}, we present our
main result, in three subsections.  
In Section~\ref{sec:weak}, we show
how to get a
Weak system, that recovers all but $k/2$ of $k$ heavy hitters, while
not increasing the noise by much.
This is a relatively slow Weak system that illustrates several
concepts, on which we build.  
in
Section~\ref{sec:sublinear}, we build on Section~\ref{sec:weak} to give
a sublinear time version of the Weak system. 
In Section~\ref{sec:toplevel}, we show
how to get a solution to the main problem.
In Section~\ref{sec:open}, we give several open
problems in connection with optimizing and generalizing our results.

\section{Preliminaries and Definitions}
\label{sec:defs}

In this section, we present notation and definitions.

\tightpgh{Systems.}
We will usually present systems for measurement and decoding as
algorithmic units without all the details of the measurement matrix
and decoding algorithm.  We will then usually argue correctness at the
system level, then argue that the system can be implemented by a matrix
with the claimed number of rows and a decoding algorithm with the
claimed runtime.

\tightpgh{Notation.}
For any vector $\signal$, we write $\signal_k$ for the best $k$-term
approximation to $\signal$ or the $k$'th element of $\signal$; it will
be clear from context.  For any vector $\signal$, we write
$\supp{\signal}$ for the {\em support} of $x$, {\em i.e.},
$\{i\,:\,\signal_i\ne 0\}$.

\tightpgh{Normalization.}
Our overall goal is to approximate $\signal_k$, the best $k$-term
approximation to $\signal$.  For the analysis in this paper, it will
be convenient to normalize $\signal$ so that
$\nerr{\signal-\signal_k}_1=1$.  It is not necessary for the decoding
algorithm
to know the original value of $\nerr{\signal-\signal_k}_1$.

\tightpgh{Heavy Hitters.}
Suppose a signal $\signal$ can be written as
$\signal=\signaly+\signalz$, where $|\supp{\signaly}|\le k$ and
$\nerr{\signalz}_1\le\eta$.  Then we say that $\supp{\signaly}$ are
the {\em $(k,\eta)$-heavy-hitters} of $\signal$.
We will frequently drop the $(k,\eta)$- when clear from context.
Ambiguity in the decomposition $\signal=\signaly+\signalz$ is inherent
in approximate sparse {\em problems} and will not
cause difficulty with our {\em algorithm}.

\tightpgh{Optimal number of measurements.}
For this paper, we only consider algorithms using the optimal number
of measurements, up to constant factors.\footnote{The optimal number
  of measurements, if $\epsilon$ is considered to be a constant,
  is~\cite{BaIndykPriceWoodruff:2010}
  $\Theta(k\log(N/k))=\Theta\left(\log\binom{N}{k}\right)$.  Our
  $\epsilon$ dependence is cubic (quadratic in a warmup algorithm),
  which is sub-optimal compared with the quadratic dependence in the
  best algorithms.}

We will use the following form of the Chernoff bound.

\begin{lemma}[Chernoff]
\label{lem:chernoff}
Fix real number $p$, $0<p<1$.  Let $X_1,X_2,\ldots,X_n$ be a set of
independent $0/1$-valued random variables with expectation $p$.
Let $X=\sum_i X_i$ and let
$\mu=pn$ denote $E[X]$.  For any $\delta>0$, we have
\[\Pr(X > (1+\delta)\mu) <
\left(\frac{e^\delta}{(1+\delta)^{1+\delta}}\right)^{\mu}
\le\left(\frac{e\mu}{a}\right)^{a},
\]
where $a=(1+\delta)\mu$.  If $a=\Omega(n)$ and $a>(1+\Omega(1))e\mu$,
then the above probability is
$p^{\Omega(n)}$.
\end{lemma}

\tightpgh{Parameter summary.}  We use the parameter $k$ for sparsity
in the toplevel signal, but a different symbol, the parameter $s$, in
subroutines, so we can say things like, ``put $s=k/2^j$ in the $j$'th
iteration.''  Similarly, the parameters $\epsilon,\alpha$, and $\eta$
are related ``noise'' or approximation ratio parameters in the various
routines, and
$\zeta$ is an ``omission'' parameter, such that we guarantee
to recover all but $\zeta s$ heavy hitters in an $s$-sparse signal.

\section{Main Result}
\label{sec:main}

\subsection{Weak System}
\label{sec:weak}

We start with a Weak System.  Intutitively, a Weak system, operating
on measurements, cuts in half the number of unknown heavy hitters,
while not increasing tail noise by much.  We estimate all values in
$\signal$ and take the largest $O(k)$ estimates.  To estimate the
values, we hash all $N$ positions into $B=O(k)$ buckets.  Each
position $i$ to estimate has a $\Omega(1)$ probability of getting
hashed into a bucket with no (other) items larger than $1/k$ and the
sum of other items not much more than the average value, $1/B\approx 1/k$,
in which case the sum of values in the bucket estimates $\signal_i$ to
within $\pm1/k$.  By a concentration of measure argument for
{\em dependent} random variables, we conclude
that $\Omega(k)$ measurements are good except with probability
$p=2^{-O(k)}$, and, if we repeat $t=O(\log(N/k))$ times, some
$\Omega(k)$ items get more than $t/2$ correct estimates except with
probability $p^t=(k/N)^{k}$.  In the favorable case, the median
estimate is correct.  The failure probability is small enough to take
a union bound over all sets of $O(k)$ positions, so we conclude that
no set of $\Omega(k)$ estimates is bad, {\em i.e.}, there are at most,
say, $k/2$ failures, as desired.

We first present an algorithm that simply estimates {\em all} $[N]$ as
suitable candidates.  This makes the runtime slightly superlinear.
Below, we will show how to get a smaller set of candidates, which
speeds the algorithm at the
cost of a controlable increase in the number of measurements.

\begin{definition}
\label{def:weak}
A
Weak system
consists of parameters $N,s,B,\eta,\zeta$, an $m$-by-$N$
{\em measurement matrix}, $\mtx{\Phi}$, and a {\em decoding algorithm},
$\mathcal{D}$.  Consider signals $\signal$ that can be written
$\signal=\mathbf{y}+\mathbf{z}$, where $|\supp{\mathbf{y}}|\le s$,
$\supp{\mathbf{y}}\subseteq I$, and
$\nerr{\mathbf{z}}_1\le O(1)$.

Given the parameters, $I$, a measurement matrix $\mtx{\Phi}$, and measurements
$\mtx{\Phi}\signal$ for any $\signal$ with a decomposition above, the
decoding algorithm returns $\widehat\signal$, such that
$\signal=\widehat\signal+\widehat{\mathbf{y}}+\widehat{\mathbf{z}}$,,
where
$|\supp{\widehat{\signal}}|\le O(s)$,
$|\supp{\widehat{\mathbf{y}}}|\le \zeta s$, and 
$\nerr{\widehat{\mathbf{z}}}_1 \le \nerr{\mathbf{z}}_1+\eta$.
\end{definition}

Without loss of generality, we may assume that
$\supp{\mathbf{y}}\cap\supp{\mathbf{z}}
=\supp{\widehat{\mathbf{y}}}\cap\supp{\widehat{\mathbf{z}}}=\emptyset$,
but, in general, $\supp{\widehat{\mathbf{x}}}$ intersects both
$\supp{\widehat{\mathbf{y}}}$ and $\supp{\widehat{\mathbf{z}}}$.

The parameter $B$ will always be set to $2s$ in implementations.  We
prove correctness for general $B$ because the generality is needed to
prove Lemma~\ref{lem:filt} below.

\begin{lemma}[Weak]
\label{lem:weak}
With probability $1-\binom{N}{s}^{-\Omega(1)}$ over the choice of hash functions,
Algorithm~\ref{algo:weak}, with appropriate instantiations of
constants, is a correct Weak system that uses
$O(\eta^{-2}\zeta^{-4}s\log(N/s))$ measurements when $B=O(s)$ and runs in time
$O(|I|\eta^{-1}\zeta^{-2}\log(N/s))$.
\end{lemma}

\begin{algorithm}
\caption{A Weak system.}
\label{algo:weak}
\begin{algorithmic}
\STATE{ Input: $N$, sparsity $s$, noise
  $\eta$, $\mtx{\Phi}$,
  $\mtx{\Phi}\signal$, hash parameter $B$,  omission
  $\zeta$, candidate set $I=[N]$}
\STATE{ Output: $\widehat\signal$}
\FOR{ $j=1$ to $t=O(\eta^{-1}\zeta^{-2}\log(N/s)/\log(B/s))$}
  \STATE{Hash $h:[N]\to[O(\eta^{-1}\zeta^{-2}B)]$}
  \FOR { $i\in I$}
    \STATE {$\signal^{(j)}_i=\sum_{h(i')=h(i)} \signal_{i'}$ // sum of signal
        values in $i$'s hash bucket---an element of input $\mtx{\Phi}\signal$}
  \ENDFOR
\ENDFOR
\FOR{ $i\in I$}
  \STATE{ Let $\signal'_i$ be the median over $j\le t$ of
    $\signal^{(j)}_i$}
\ENDFOR
\STATE{ Zero out all but the largest $O(s)$ elements of $\signal'$; get
  $\widehat\signal$}
\RETURN{ $\widehat\signal$} 
\end{algorithmic}
\end{algorithm}

\begin{proof}
The number of measurements and runtime are as claimed by construction,
so we show correctness.  There are several parts to this, and much of
this is similar to or implicit in previous work.  We show that, with
probability at least 3/4 over the choice of $\mtx\Phi$:
\begin{enumerate}
\item\label{en:collision} For any set $S=\supp{\signaly}$ of $s$ heavy hitters and any
  set $D=\supp{\widehat{\signal}}$ of
  $s$ ``decoys'' that might displace $S$, at
  most $O(\zeta s)$ elements of $S\cup D$ collide, in at least $t/4$ of
  their buckets, with an element of
  $S\cup D\cup T$, where $T$ is the set of the top
  $O(s/(\zeta\eta))$
  elements.
\item\label{en:flat} Let $A$ be the set of rows of $\mtx{\Phi}$ with a
  one anywhere in
  columns $S\cup D$ and let $\mtx{\Phi}_A$ be $\mtx{\Phi}$ restricted to the
  rows of $A$.  Let $F$ be a set of $\Omega(s/(\zeta\eta))$ columns
  disjoint from $S\cup D\cup T$, and let
  $\nu$ be an $N$-vector such that
  $\nu=1/|F|$ on $F$ and zero elsewhere ($\nu$ is a ``flat tail'').
  We have
  $\nerr{\mtx{\Phi}_A \nu}_1\le O(\eta\zeta t)$.
\item\label{en:tail} Let $A$, $\mtx{\Phi}_A$.
  For any $\nu$ supported off $S\cup D$
  with $\nerr{\nu}_1=1$ and
  $\nerr{\nu}_\infty \le O(1/|\supp{T}|)=O(\eta\zeta/s)$,
  we have
  $\nerr{\mtx{\Phi}_A \nu}_1\le O(\eta\zeta t)$.
\item\label{en:est} There is a decomposition
  $\mathbf{x}=\widehat{\mathbf{x}'}+\widehat{\mathbf{y}'}+\widehat{\mathbf{z}'}$
  such that:
\begin{itemize}
\item $\widehat{\mathbf{x}'}$ equals $\mathbf{x}'$ on
  $\supp{\mathbf{x}_s}$ and zero elsewhere, where $\mathbf{x}'$ is as in Algorithm~\ref{algo:weak},
\item $|\supp{\widehat{\mathbf{y}'}}|\le\zeta s$
\item $\nerr{{\widehat{\mathbf{z}'}}}_1\le\nerr{\mathbf{z}}_1+O(\eta)$.
\end{itemize}
\item\label{en:conc} (The lemma's conclusion.)  There's a decomposition
$\signal=\widehat\signal+\widehat{\mathbf{y}}+\widehat{\mathbf{z}}$
with
$|\supp{\widehat{\signal}}|\le O(s)$,
$|\supp{\widehat{\mathbf{y}}}|\le \zeta s$, and 
$\nerr{\widehat{\mathbf{z}}}_1 \le \nerr{\mathbf{z}}_1+\eta$.
\end{enumerate}

\iffull
The dependence is as follows.  Item~\ref{en:tail} for general tails
follows from Item~\ref{en:flat} for flat tails.  Item~\ref{en:est}
follows from Items~\ref{en:collision} and~\ref{en:tail} and shows that
the {\em estimates} lead to an acceptable decomposition of $\signal$,
assuming {\em some} choice (generally unknown to the algorithm) of
support for $\widehat{\mathbf{x}}$, namely $\supp{\signal_s}$.
Finally, Item~\ref{en:conc} follows from Item~\ref{en:est} by
considering the displacement of an element in the support of $\signal_s$
by an element in the Algorithm's output, {\em i.e.}, the support of
$\widehat{\mathbf{x}}$.  Only Items~\ref{en:collision}
and~\ref{en:flat} involve probabilistic arguments.

\paragraph{Item~\ref{en:collision}.}
Fix a decomposition $\signal=\mathbf{y}+\mathbf{z}$ as
above, let $S$ equal $\supp{\mathbf{y}}$, and let $D\subseteq[N]$ be
any set of $s$ positions.  (We only care about the case
$D=\supp{\widehat\signal}$, but, to handle stochastic dependence
issues, it is
necessary to prove the result for a general $D$ of this size.)
We want to
show that at most $O(\zeta s)$ elements of $S\cup D$ collide with
one of the top $O(s/(\zeta\eta))$ elements in at least $t/4$ of their
$t$
buckets.  Let $T$ be the set of top $O(s/(\zeta\eta))$ elements in
$[N]$.

Intuitively, there are
$\Omega(\eta^{-1}\zeta^{-2}B)$
hash buckets and at most
$O(|T|)$ are ever occupied by an element of $S\cup D\cup T$, so each
element of $S\cup D$ has at most a
$O(T\eta\zeta^2/B)=O(s\zeta/B)\le O(\zeta)$ chance to collide when it
is hashed.  As we discuss below,
this
implies that the expected number of collisions (at the time of hashing
{\em or later}) is $O(s\zeta/B)$
in each of the $t$ repetitions.
If all estimates (over all $i$ and all
repetitions) were independent, we could apply the Chernoff bound
Lemma~\ref{lem:chernoff}, and conclude that the number of failed
element-repetition pairs exceeds $O(\zeta |S\cup D|t)=O(\zeta st)$ only
with probability $\binom{N}{|T|}^{-\Omega(1)}$, small enough to take a
union bound over all $(S,D,T)$, which is acceptably small.  But it is
easy to see (and also see
below) that there is at least some small dependence.  So instead we
proceed as follows, using a 
form of the
Method of Bounded Differences and coupling~\cite{DP,MR,MU}.

First hash the elements of $T\setminus(S\cup D)$.  Then hash the
elements of $S\cup D$,
in some arbitrary order.  Let $X_j$ be the
0/1-valued random variable that takes the value 1 if the $j$'th
element of $S\cup D$ is hashed into a bucket that is bad (occupied by
an element of $S\cup D\cup T$) at
the time of $j$'s
hashing.  As above, each $X_j$ has $E[X_j]\le \zeta$.

Note that even if some $i\in S\cup D$ is isolated at the time of its
hashing, $i$ may become clobbered by an element of $j\in S\cup D$
that is later hashed into its bucket.  So $\sum_j X_j$ is {\em not}
the total number of failed estimates.  But observe that if some $j$ is
hashed into the same bucket as previously-hashed items, it can only
clobber at most one other previously-unclobbered element $i$, because $j$ is
only hashed into one bucket, and that bucket has at most one
previously-unclobbered item.  It follows that $2\sum_j X_j$ is an
upper bound on the number of colliding items in $S\cup D$, where,
for some $p$, the $X_j$'s are 0/1-valued random variables with the
expectation of each $X_j$ bounded by $p$, even conditioned on any
outcomes of $X_{<j}$.
This is enough to get the conclusion of the
Chernoff inequality with independent trials of failure probability
$p$, by a standard coupling argument.  (See, {\em e.g.}, exercise 1.7
of~\cite{DP}.)
In the standard proof of Chernoff, we have, for any $\lambda>0$,
\begin{eqnarray*}
\Pr\left( \sum_{j=1}^n X_j \ge a\right)
&  =  & \Pr\left( e^{\lambda \sum X_j} \ge e^{\lambda a}\right)\\
&  =  & \Pr\left( \prod e^{\lambda X_j} \ge e^{\lambda a}\right)\\
& \le & E\left[\prod e^{\lambda X_j}\right]/e^{\lambda a}.
\end{eqnarray*}
At this point, if the $X_j$'s were independent, we would get the
product of expectations.  Instead, we proceed as follows, where
$Y_j$'s are independent random variables with expectation $p$.
\begin{eqnarray*}
\Pr\left( \sum X_j \ge a\right)
& \le & E\left[\prod e^{\lambda X_j}\right]/e^{\lambda a}\\
&  =  & E\left[e^{\lambda X_n}\prod_{j<n} e^{\lambda X_j}\right]/e^{\lambda a}\\
&  =  & \sum_{\vec v}(\Pr\left(X_n = 0|X_{<n}=\vec v\right) + \Pr\left(X_n = 1|X_{<n} = \vec v\right)e^\lambda)\\
&     & \quad\cdot
        \Pr( X_{<n} = \vec v)
        e^{\lambda\cdot\text{weight}(\vec v)}/e^{\lambda a}\\
& \le & \sum_{\vec v}(1 - p + pe^\lambda) \Pr( X_{<n} = \vec v)
        e^{\lambda\cdot\text{weight}(\vec v)}/e^{\lambda a}\\
& = & E\left[e^{\lambda Y_n}\prod_{j<n} e^{\lambda X_j}\right]/e^{\lambda a}.
\end{eqnarray*}
Proceed inductively, getting
\[\Pr\left( \sum X_j \ge a\right)\le E\left[\prod e^{\lambda Y_j}\right]/e^{\lambda a}
  =\frac{\prod E\left[e^{\lambda Y_j}\right]}{e^{\lambda a}},\]
to which the rest of the usual Chernoff-type bound applies.
Thus the expected number of pairs of elements
in $S\cup D$ and repetition that collide is at
most $O(\zeta st)$.

Having shown that our dependent collision events behave like
independent events up to constants, we now go over the arithmetic,
assuming independent collisions.  Each $i\in S\cup D$ fails in each
repetition with probability at most $O(\zeta s/B)$ (wlog, exactly
$\zeta s/B$ for now).  Among the $(2st)$ pairs of $i\in S\cup D$ and
repetition, we expect to get $\mu=O(\zeta s^2t/B)$ failed pairs, and
we get at least $a\ge \zeta st$ failures with
probability at most $\left(e\mu/a\right)^a$, by
Lemma~\ref{lem:chernoff}, Chernoff.  So the
failure probability is
\[\left(e\mu/a\right)^a
=(s/B)^{\Omega(\zeta s t)}
=(s/B)^{\Omega(\eta^{-1}\zeta^{-1}s\log(N/s)/\log(B/s))}
=(s/N)^{\Omega(\eta^{-1}\zeta^{-1}s)},\]
which is small enough to to take a union bound over $(T,S,D)$.  In the
favorable
case, there is only a fraction $O(\zeta)$ of all pairs of item and
repetition with a failed estimate.  It follows, after adjusting
constants, that less than $(1/2)\zeta s$ items get more that $t/4$ failed
original estimates.  The remaining $(1-\zeta/2)s$ items get good final
median estimate (even if another $t/4$ original estimates fail for
other reasons, as we discuss below), since a median estimate fails
only if a majority of
mediand estimates fail.

\paragraph{Item~\ref{en:flat}.}  Fix $S,D,T,|F|,F$, choose the $S\cup D$
columns and the $T$ columns of $\mtx\Phi$ (arbitrarily for this
discussion), and
thereby define $A$ (the rows of $\mtx\Phi$ with a 1 in columns $S\cup
D$) and $\nu$ (equal to $1/|F|$ on $F$ and zero elsewhere), as above.
We now hash the elements of
$F$ at random, {\em i.e.}, choose the $F$ columns of $\mtx\Phi$.  In
each repetition, there are 
$O(\eta^{-1}\zeta^{-2}B)$ buckets, of which $O(s)$ are in $A$.
It follows that each element in $F$ hashes to $A$ in
each repetition with probability $O(\eta\zeta^2 s/B)$.  Counting repetitions,
there are a total of $t|F|$ elements that each hash into $A$ with
probability $\eta\zeta^2 s/B$.   We expect
$\mu = \eta\zeta^2 t|F|s/B$ element-repetition pairs of $t|F|$ total to
hash into $A$ and we 
get more than $a=\eta\zeta^2 t|F|$ with probability at most
\[(e\mu/a)^a
= (s/B)^{\Omega(\eta\zeta^2 t|F|)}
\le (s/B)^{\Omega(|F|\log(N/s)/\log(B/s))}
= (s/N)^{\Omega(|F|)},
\]
which is small enough to take a union bound over all $S,D,T,|F|,F$.
Since elements of $\nu$ have magnitude $1/|F|$, it follows that
$\nerr{\mtx\Phi_A\nu}_1\le a/|F|=O(\eta\zeta^2 t)$, so\footnote{This
  seems loose by a factor $\zeta$, but local fixes, like replacing
  $\zeta$ with
  $\sqrt{\zeta}$, do not seem to work.  We speculate that better
  dependence on $\zeta$ is possible.} we conclude
$\nerr{\mtx\Phi_A\nu}_1\le O(\eta\zeta t)$.

At this point, we have that, except with probability $1/4$, at most
$O(\zeta s)$ of $S\cup D$
items collide with $S\cup D$ or with an element of $T$ (of magnitude
at least $\eta\zeta/s$) in more than $t/4$ of their repetitions and no
flat tail of support size at least
$s/(\eta\zeta)$ contributes more than a constant times its expected
amount, which is
$O(\eta\zeta t)$ if the magnitude of $\nu$ is maximal, into the buckets $A$
containing the top $s$ heavy
hitters.  Conditioned on this holding, we proceed
non-probabilistically.

\paragraph{Item~\ref{en:tail}.}  Let $\nu$ be any vector supported
disjointly from $S\cup D$ with $\nerr{\nu}_1=1$ and
$\nerr{\nu}_\infty\le O(\zeta\eta/s)$ as above.
Since $\mtx\Phi$ is non-negative, we may assume that $\nu$ is
non-negative, as well, by replacing $\nu$ with $|\nu|$.  Next, round
each non-zero element of $\nu$ up to the nearest power of 2, at most doubling
$\nu$.  Write $\nu=\sum_i w_i\nu_i$, where $\nu_i$ takes on only the
values $0$ and $2^{-i}$, and $w_i$ is 0 or 1.
Also write $\nu = \nu'+\nu''$, where $\nu_i$ contributes to $\nu'$ if
the support of $\nu_i$ is at least $s/(\eta\zeta)$ and $\nu_i$
contributes to $\nu''$, otherwise.
The $\nu_i$'s contributing to $\nu'$ are multiples of flat tails of
the kind handled in
Item~\ref{en:flat} and their sum, $\nu'$, which has 1-norm at most 1, is
a subconvex combination of such flat tails.  Since $\nerr{\mtx\Phi_A\nu}_1$
is subadditive in $\nu$ (actually, strictly additive under our
non-negativity assumption), we get
$\nerr{\mtx\Phi_A\nu'}_1\le O(\eta\zeta t)$.

Now consider the sum $\nu''$ of $\nu_i$ with support less than
$s/(\eta\zeta)$.  In general, these {\em can} contribute more than
their expected value, but not {\em much} more than the expected value,
and the expected value is typically much less than for other flat
tails.  We will handle the sum of these at once (without using the
convex combination argurment), so we may assume the supports are the
maximum, $s/(\eta\zeta)$, by increasing each actual support to a
superset.  Also, we may assume that the corresponding $w_i$'s are as
large as possible, {\em i.e.}, $w_i=1$ if $2^{-i}\le \eta\zeta/s$ and
$w_i=0$, otherwise (so that the maximum magnitude is $\eta\zeta /s$).
With these assumptions, each such flat tail
contributes not much more than its expected number, $O(st)$, of elements
of magnitude $2^{-i}=2^{-j}\eta\zeta/s$ for some $j\ge 0$.  Thus
$\nerr{\mtx\Phi_A \nu_i}_1 = O(\eta\zeta t2^{-j})$ for $i$ and $j$ as above.
The {\em sum} (which can be greater than a convex combination of
the original contribution but, it turns out, is at most a constant
times a
convex combination under our assumptions) contributes
$\nerr{\mtx\Phi_A\nu''}_1 \le O(\eta\zeta t\sum_{j\ge 0} 2^{-j})
=O(\eta\zeta t)$,
as desired.

Thus $\nerr{\mtx\Phi_A\nu}_1\le
\nerr{\mtx\Phi_A\nu'}_1+\nerr{\mtx\Phi_A\nu''}_1
\le O(\eta\zeta t)$.

\paragraph{Item~\ref{en:est}.}  Let $\widehat{\mathbf{x}}'$ be as
above.  In Item~\ref{en:collision}, we showed that an acceptable
number $O(\zeta s)$ elements of $S\cup D$ suffer collisions; here we
we consider only the elements of $S\cup D$ that do not collide with
$S\cup D\cup T$.  So we can consider only the tail elements
that are still relevant, {\em i.e.}, the elements of
$[N]\setminus(S\cup D\cup T)$, which have magnitude at most
$\eta\zeta/s$.  These form a tail $\nu$ as described in
Item~\ref{en:tail}.  Consider $i$ to be a failure if
$|\widehat{\signal}'_i-\signal_i|\ge \Omega(\eta/s)$.  Then each
failed $i$ in $\widehat{\mathbf{x}}'$ requires $t/2$ failed
$i$'s in $\signal^{(j)}$'s and, since collisions only account for
$t/4$ $i$'s in $\signal^{(j)}$'s, each failed $i$ in
$\widehat{\mathbf{x}}'$ that does not fail due to collisions also
requires $\Omega(t)$ failed $i$'s in $\signal^{(j)}$'s.  Thus each
failed but non-colliding $i$ accounts for $\Omega(t \eta/s)$ of
$\nerr{\mtx\Phi_A\nu}$.  Since
$\nerr{\mtx\Phi_A\nu}\le O(\eta\zeta t)$, there can be at most
$O(\zeta s)$ failures, as desired.  The remaining at-most-$s$ estimates
of $\signal_s$ each are good to within $O(\eta/s)$, additively, so the
total 1-norm of the estimation errors is $O(\eta)$, as desired.

\paragraph{Item~\ref{en:conc}.}
To complete our analysis of correctness, we describe
$\widehat\signal,\widehat{\mathbf{y}}$, and
$\widehat{\mathbf{z}}$ and show that they have
the claimed properties.  This is summarized in
Table~\ref{table:decomp}.

\begin{table}
\caption{Contributions $\signal_i-\widehat\signal_i$ to
  $\widehat{\mathbf{y}}$ and to $\widehat{\mathbf{z}}$ from
  $i\in\supp{\mathbf{y}}$ and
  $i\in\supp{\mathbf{z}}$, according to whether $i\in\widehat{\signal}$,
  whether $i\in\supp{\signal_i}$ has a good or bad estimate ({\em i.e.} whether
  or not the median
  estimate is good to within $\pm O(\eta/s)$), or,
  if $i\in\supp{\mathbf{y}}\setminus\widehat{\mathbf{\signal}}$, according to
  whether $i$ was displaced by $i'$ with a good or bad estimate, under
  an arbitrary pairing between
  $i\in\supp{\mathbf{y}}\setminus\supp{\widehat\signal}$ and 
  $i'\in\supp{\widehat{\signal}}\setminus\supp{\mathbf{y}}$.  Note that
  zero may be a good estimate.}
\label{table:decomp}
\medskip
\begin{center}
\begin{tabular}{|r|c|c|c|c|}
\cline{2-5}
\multicolumn{1}{c|}{~}
& \multicolumn{2}{c|}{$i\in\supp{\mathbf{y}}$}
& \multicolumn{2}{c|}{$i\in\supp{\mathbf{z}}$} \\
\cline{2-5}
\multicolumn{1}{c|}{~}
& Good     & Bad      & Good     & Bad     \\
\multicolumn{1}{c|}{~}
& estimate & estimate & estimate & estimate\\
\cline{2-5}
\hline
{$i\in\supp{\widehat\signal}$} &
  $\widehat\signalz$ & $\widehat\signaly$ & $\widehat\signalz$ & $\widehat\signaly$ \\
\hline
{$i\notin\supp{\widehat\signal}$; Displaced by bad estimate} &
  \multicolumn{1}{c}{$\widehat\signaly$} &
  $\widehat\signaly$ & 
  \multicolumn{1}{c}{$\widehat\signalz$} & $\widehat\signalz$ \\
\cline{1-2}
{$i\notin\supp{\widehat\signal}$; Displaced by good estimate} &
  $\widehat\signalz$ & $\widehat\signaly$ & 
  \multicolumn{1}{c}{$\widehat\signalz$} & $\widehat\signalz$ \\
\hline
\end{tabular}
\end{center}
\end{table}

\begin{itemize}
\item The pseudocode Algorithm~\ref{algo:weak} returns
  $\widehat\signal$, which has support size $O(s)$.
\item Elements $i\in\supp{\widehat\signal}$ with a good estimate (to
  within $\pm O(\eta/s)$) contribute $\signal_i-\widehat\signal_i$ to
  $\widehat{\mathbf{z}}$.  There are at most $O(s)$ of these, each
  contributing $O(\eta/s)$, for total contribution $O(\eta)$ to
  $\widehat{\mathbf{z}}$.
\item Elements $i\in\supp{\widehat\signal}$ with a bad estimate (not
  to within $\pm O(\eta/s)$) contribute $\signal_i-\widehat\signal_i$ to
  $\widehat{\mathbf{y}}$.  There are at most $O(\zeta s)$ of these.
\item Elements $i\in\supp{\mathbf{z}}\setminus\supp{\widehat\signal}$
  contribute $\signal_i$ to $\widehat{\mathbf{z}}$.  The $\ell_1$ norm of these is
  at most $\nerr{\mathbf{z}}$.
\item Elements $i\in\supp{\mathbf{y}}\setminus\supp{\widehat\signal}$
  with a good estimate that are nevertheless displaced by another
  element $i'\in\supp{\widehat\signal}\setminus\supp{\mathbf{y}}$ with
  a good estimate contribute to $\widehat{\mathbf{z}}$.
  There are at most $s$ of these.  While the value $\signal_i$ may be
  large and make a large contribution to $\widehat{\mathbf{z}}$, this
  is offset by $\signal_{i'}$ satisfying, for some $c$,
$|\signal_{i'}|
\ge |\widehat{\signal}_{i'}|-c\eta/s
\ge |\widehat{\signal}_{i}|-c\eta/s
\ge |\signal_{i}|-2c\eta/s$, which 
  contributes to $\mathbf{z}$ but {\em not} to
  $\widehat{\mathbf{z}}$.  Thus the net contribution to
  $\widehat{\mathbf{z}}$ is at most $O(\eta/s)$ for each of the $O(s)$
  of these $i$, for
  a total $O(\eta)$ contribution to $\widehat{\mathbf{z}}$.

  The contributions of such $i$ and $i'$ are summarized in the
  following table, whence the reader can confirm that
  $(\signaly+\signalz)_{\{i,i'\}}
  =(\widehat\signal+\widehat\signaly+\widehat\signalz)_{\{i,i'\}}$ and
  $\nerr{\widehat\signalz_{\{i,i'\}}}\le\nerr{\signalz_{\{i,i'\}}}+O(\eta/s)$.

\begin{center}
\begin{tabular}{|c||c|c|c|c|c|}
\hline
 & $\signaly$ & $\signalz$ & $\widehat\signal$ & $\widehat\signaly$ &
$\widehat\signalz$ \\
\hline
\hline
$i$ & $\signal_i$ & & & & $\signal_i$ \\
\hline
$i'$ & & $\signal_{i'}$ & $\widehat{\signal}_{i'}$ & &
$\signal_{i'}-\widehat{\signal}_{i'}$ \\
\hline
\end{tabular}
\end{center}

\item Elements $i\in\supp{\mathbf{y}}\setminus\supp{\widehat\signal}$
  that themselves have bad estimates or are displaced by elements with
  bad estimates contribute $\signal_i$ to $\widehat{\mathbf{y}}$.  There are at
  most $\zeta s$ bad estimates overall, so there are at most $O(\zeta s)$
  of these.
\end{itemize}

We have shown that $|\supp{\widehat\signaly}|\le O(\zeta s)$ and
$\nerr{\widehat\signalz}_1\le\nerr{\signalz}_1+O(\eta)$.  By adjusting
constants in the algorithm, we can arrange for the conclusion of the Lemma.

\fi
\end{proof}

\subsection{Sublinear Time}
\label{sec:sublinear}

In this section, we introduce a way to limit $I$ to get a sublinear time
Weak system.  Since the runtime of the weak system will dominate
the overall runtime, it  follows that the overall algorithm will have
sublinear time.  We first give a basic algorithm with runtime
approximately $\sqrt{kN}$, then we generalize from
$\sqrt{kN}=k(N/k)^{1/2}$ to $\ell^{O(1)}k(N/k)^{1/\ell}$ for any
positive integer $\ell$, but with number of measurements suboptimal by
the factor $\ell^{O(1)}$.

The basic idea, for $\ell=2$ and (ignoring for now the small effects
of $\epsilon$ that we set to $\Omega(1)$), is as follows.  Hash
$h:[N]\to[\sqrt{kN}]$, and repeat a total of two times.  In each
repetition, a heavy hitter avoids collisions except with probability
$k/\sqrt{kN}=\sqrt{k/N}$.  Also, the average amount of tail noise (sum
of others in the bucket) is
$1/\sqrt{kN}$, so the tail noise exceeds $1/k$ on at most the fraction
$k/\sqrt{kN}=\sqrt{k/N}$ of the buckets.   So a heavy hitter dominates
its bucket except with probability $O(\sqrt{k/N})$.  The heavy
hitter dominates in at least one of the two repetitions with failure
probability equal to the square of that, or $O(k/N)$, which is what we
would need
to apply the Chernoff bound and to conclude
that, except with probability $\binom{N}{k}^{-1}$ (which is small
enough to take our union bound), $\Omega(k)$ of
the heavy hitters are isolated in low-noise buckets.  There is some
dependence here, which is handled as in Section~\ref{sec:weak}.

Now focus on one of the two repetitions.
We can form a new signal $x'$ of length $N'=\sqrt{kN}$ and sparsity
$k'=\Omega(k)$.  The signal $x'$ is
indexed by hash buckets and $x'_j=\sum_{h(i)=j} x_i$, {\em i.e.}, we
sum the values in $x$ that are hashed to the same bucket.  The
original $(N,k)$ signal (of length $N$ and
sparsity $k$) and a new $(N',k')=(\sqrt{kN},\Omega(k))$ signal form what
we call a two-level {\em signal filtration}, of which there are two,
for the two repetitions.

For each filtration, run the Weak system Algorithm~\ref{algo:weak} on
the $(N',k')$
signal $\signal'$, getting a set $H$ of $\Theta(k)$ heavy hitters.  This
uses $O(k'\log(N'/k'))=O(k\log(N/k))$ measurements and runtime led by
the factor $N'=\sqrt{kN}$.  Form the
set $I=h^{-1}(H)$ of indices to the {\em original} signal.  Finally,
run the Weak system on the original
signal, but with index set $I$.  This also takes $O(k\log(N/k))$
measurements and runtime led by the factor $|I|=\sqrt{kN}$.  Thus the
overall runtime is given by the time to make two exhaustive searches
over spaces of size about $\sqrt{kN}$, on each of two repetitions,
{\em i.e.}, $\ell$ repetitions of $\ell$ exhaustive searches over
spaces of size $k(N/k)^{1/\ell}$, for $\ell=2$.  For correctness, we
need to argue that the filtration is faithful to the original signal
in the sense that enough heavy hitters from the original signal become
heavy hitters in the $(k',N')$ signals and that we can successfully
track enough of these back to the original signal.

In the general situation, $\ell$ may be greater than 2.  We will have
$\ell-1$ intermediate signals in the levels of the filtration, which we
define below.  The runtime will arise from performing $\ell$
repetitions of $\ell$
cascaded exhaustive searches over spaces of size about
$k(N/k)^{1/\ell}$.  There is strong overlap between the set of heavy
hitters in the original signal and the set of heavy hitters in the
shortest signal (of length $k(N/k)^{1/\ell}$).  Assuming a
correspondence of heavy hitters, our task is to trace each such heavy
hitter in the shortest signal through
longer and longer signals, back to the original $(N,k)$ signal.
Unfortunately, each time we ascend a level, we encounter more noise,
and risk losing the trail of our heavy hitter.  In the case of general
$\ell$, we will need to control noise and other losses by setting
parameters as a function of $\ell$.  Roughly speaking, we need to lose
no more than
about $k/\ell$ heavy hitters at each level {\em i.e.},
$|\supp{\widehat{\signaly}}|\le k/\ell$, rather than losing, say,
$k/2$, and (for general $\epsilon$) we need to increase the noise by
at most $O(\epsilon/\ell)$ rather
than $O(\epsilon)$, {\em i.e.}, $\nerr{\widehat\signalz}_1-\nerr{\signalz}_1\le
O(\epsilon/\ell)$.    This is done by setting the parameter $\zeta$ to $1/\ell$
and $\eta$ to $O(\epsilon/\ell)$ instead
of $O(\epsilon)$.  Also, the number of repetitions must increase
from $O(\ell)$ to $O(\ell/\epsilon)$.

We now proceed formally, for general number $\ell$ of levels.

\begin{definition}
Fix integer parameters $s,N$, and $\ell$, and real $\xi>0$.  Given a signal $\signal$
and a hash function $h:[N]\to[O((s/\xi)(N/s)^{1/\ell})]$,
an $\ell$-level {\em signal filtration} on $\signal$ is a collection
of $\ell$ signals,
$\signal^{(1)},\signal^{(2)},\ldots, \signal^{(\ell)}$, defined as
follows.  The signal $\signal^{(q)}$ has length
$N^{(q)}=O\left((s/\xi)(N/s)^{q/\ell}\right)$.  Use the hash function
$h:[N]\to[N^{(1)}]$ and
define $\signal^{(1)}_j$ by $\signal^{(1)}_j=\sum_{h(i)=j}\signal_i$.
Then, for $1\le q<\ell$, define $\signal^{(q+1)}$ from $\signal^{(q)}$ by
splitting each
subbucket $b$ indexing an element of $\signal^{(q)}$ ({\em i.e.,} a
subset of $[N]$) into subsubbuckets, in some arbitrary,
deterministic way.  Denote by $\mysplit{b}$ the resulting set of
subsubbuckets.  Then $\signal^{(q+1)}=\bigcup_b\mysplit{b}$.  Each
subbucket is split into exactly $(N/s)^{1/\ell}$ subsubbuckets
except
that buckets in $\signal^{(\ell-1)}$, which have size only
$\xi (N/s)^{1/\ell}$, are split into $\xi(N/s)^{1/\ell}$
singletons, resulting in
$\signal$.  See Figure~\ref{fig:filtration}.
\end{definition}

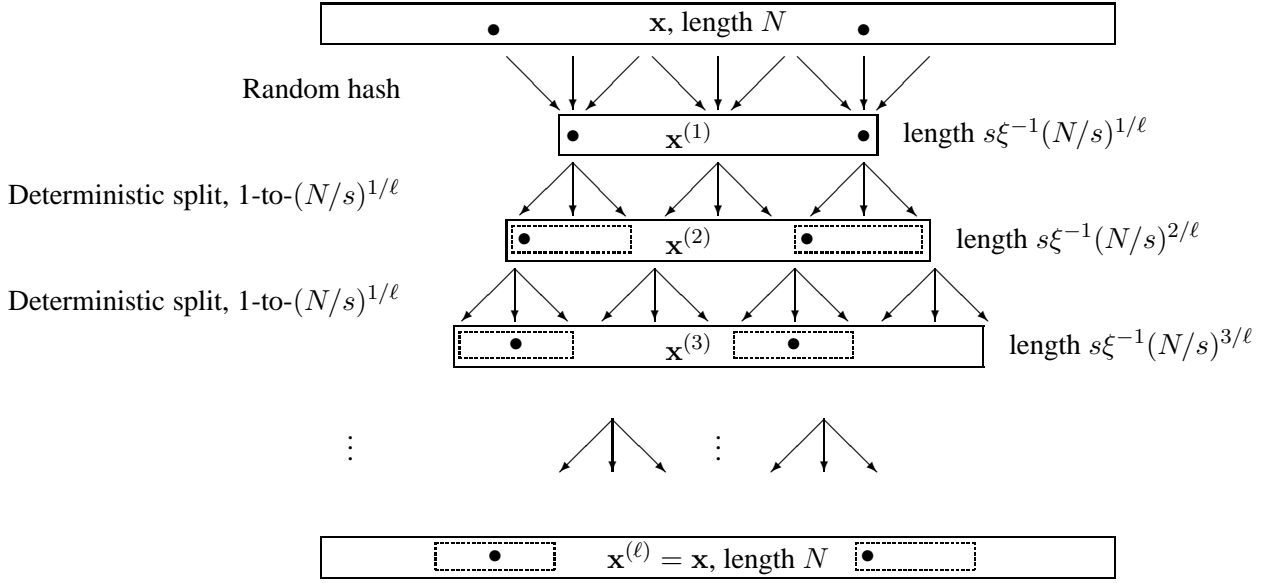
\begin{figure}[t]
\caption{A signal filtration.  Heavy hitters, denoted by bullets, are likely
  isolated in low-noise buckets by the hashing, in which case they
  dominate their buckets at all levels of the deterministic
  splitting.  Algorithm~\ref{algo:weak} (Weak) is used to search all of
  $\signal^{(1)}$.  For $q>1$, given a set $H$ of heavy hitters in
  $\signal^{(q)}$, Algorithm~\ref{algo:weak} (Weak) is also used to find
  heavy hitters in $\signal^{(q+1)}$, but we search only
  $(N/s)^{1/\ell}$ items of $\signal^{(q+1)}$ in $I=\bigcup_{b\in H}\mysplit{b}$
  (indicated by dashed boxes).}
\label{fig:filtration}
\begin{center}
\begin{picture}(300,220)(-180,0)

\put(-150,202){\framebox(300,15)[c]{$\signal$, length $N$}}

\put(-85,207){\makebox(0,0)[c]{$\bullet$}}
\put( 55,207){\makebox(0,0)[c]{$\bullet$}}

\multiput(-80,197)(55,0){3}{
\put(0,0){\vector( 1,-1){20}}
\put(25,0){\vector( 0,-1){20}}
\put(50,0){\vector(-1,-1){20}}
}

\put(-120,185){\makebox(0,0)[r]{Random hash}}

\put( -60,160){\framebox(120,15)[c]{$\signal^{(1)}$\quad\quad}}
\put(  70,162){\makebox(0,0)[bl]{length $s\xi^{-1}(N/s)^{1/\ell}$}}

\put(-55,167){\makebox(0,0)[c]{$\bullet$}}
\put( 55,167){\makebox(0,0)[c]{$\bullet$}}

\multiput(-55,157)(55,0){3}{
\put(0,0){\vector(-1,-1){20}}
\put(0,0){\vector( 0,-1){20}}
\put(0,0){\vector( 1,-1){20}}
}

\put(-120,145){\makebox(0,0)[r]{Deterministic split, 1-to-$(N/s)^{1/\ell}$}}

\put( -80,120){\framebox(160,15)[c]{$\signal^{(2)}$\quad\quad}}
\put( +90,122){\makebox(0,0)[bl]{length $s\xi^{-1}(N/s)^{2/\ell}$}}


\put(-78,123){\dashbox(45,10)[l]{\,$\bullet$}}
\put( 29,123){\dashbox(48,10)[l]{\,$\bullet$}}

\multiput(-77,117)(53,0){4}{
\put(0,0){\vector(-1,-1){20}}
\put(0,0){\vector( 0,-1){20}}
\put(0,0){\vector( 1,-1){20}}
}

\put(-120,105){\makebox(0,0)[r]{Deterministic split, 1-to-$(N/s)^{1/\ell}$}}

\put(-100, 80){\framebox(200,15)[c]{$\signal^{(3)}$\quad\quad}}
\put( 110, 82){\makebox(0,0)[bl]{length $s\xi^{-1}(N/s)^{3/\ell}$}}


\put(-98,83){\dashbox(43,10)[c]{$\bullet$}}
\put(  6,83){\dashbox(45,10)[c]{$\bullet$}}

\put(   0, 52){\makebox(0,0)[c]{$\vdots$}}
\put(-150, 52){\makebox(0,0)[c]{\quad\quad$\vdots$}}

\multiput(-40,60)(80,0){2}{
\put(0,0){\vector(-1,-1){20}}
\put(0,0){\vector( 0,-1){20}}
\put(0,0){\vector( 1,-1){20}}
}

\put(-150,  0){\framebox(300,15)[c]{$\signal^{(\ell)}=\signal$, length $N$}}


\put(-107,3){\dashbox(45,10)[c]{$\bullet$}}
\put(  52,3){\dashbox(45,10)[l]{\,$\bullet$}}

\end{picture}
\end{center}
\end{figure}

Consider a heavy index $i$ in the original signal.  It maps to a
bucket, $h(i)$.  In the favorable case, $i$ dominates $h(i)$, in the
sense that $|\signal_i|$ accounts for, say, $3/4$ of the $\ell_1$ norm
of $h(i)$.  Because the rest of the filtration involves only splitting buckets, it follows that
$i$ will dominate its bucket at each level of the filtration.  For
sufficiently many such $i$'s, we therefore find the bucket the
containing $i$ in level $q+1$ using a Weak algorithm, inductively
assuming we had the correct bucket at level $q$.  We first show that
enough heavy $i$'s dominate their buckets.

\begin{lemma}[Filtration Hashing]
\label{lem:filt}
Fix parameters $N,s,\ell,\alpha$ and let $\xi=\Theta(\alpha)$.  Let
\[h_j:[N]\to[(s/\xi)(N/s)^{1/\ell}]\]
be $O(\ell/\alpha)$ independent hash
functions.    
With adjustable probability $\Omega(1)$ over $\mtx{\Phi}$, the following
holds.  Given signal $\signal$, suppose
$\signal=\mathbf{y}+\mathbf{z}$,
with $|\supp{\mathbf{y}}|\le s$ and
$\nerr{\mathbf{z}}_1=1$, and suppose, without loss of generality, that
$|\signal_i|\ge \Omega(\alpha/s)$ for $i\in\supp{\mathbf{y}}$.  We have
$\signal=\widehat{\mathbf{y}}+\widehat{\mathbf{w}}+\widehat{\mathbf{z}}$,
where, for all $i\in\supp{\widehat{\mathbf{y}}}$, $i$ dominates some
$h_j(i)$ and $|x_i|\ge \alpha/s$,
$|\supp{\widehat{\mathbf{w}}}|\le s/6$, and $\nerr{\mathbf{z}}_1\le
1+O(\alpha)$.
\end{lemma}

\begin{proof}
This follows directly from Lemma~\ref{lem:weak}, Item~\ref{en:est},
letting $\zeta$ be a constant, the $B$ of Lemma~\ref{lem:weak} equal
$s(N/s)^{1/\ell}$, and $\eta$ of Lemma~\ref{lem:weak} equal
$\Theta(\xi)$ (which is also $\Theta(\alpha)$).
Then $\mathbf{x}'$ of Lemma~\ref{lem:weak}, Item~\ref{en:est} gives
$\widehat{\mathbf{y}}$ of this lemma (these are the surviving heavy
hitters);
$\widehat{\mathbf{y}}$ of Lemma~\ref{lem:weak} gives
$\widehat{\mathbf{w}}$ of this lemma (these are the ruined heavy 
hitters), and the $\widehat{\mathbf{z}}$'s in the Lemmas coincide.
\end{proof}

Our Sublinear Time Weak system is given in
Algorithm~\ref{algo:sublinear}.

\begin{algorithm}
\caption{A Fast Weak system.}
\label{algo:sublinear}
\begin{algorithmic}
\STATE{ Input: $N$, sparsity $s$, noise $\alpha$,
  $\mtx{\Phi}$, $\mtx{\Phi}\signal$}
\STATE{ Global integer $\ell\ge 2$ // optimize for application}
\STATE{ Output: $\widehat\signal$}
\FOR{ $j\leftarrow 1$ to $t=O(\ell/\alpha)$}
  \STATE{ Pick hash function $h:[N]\to[N^{(1)}]$, using parameters
    $N,s,\ell$ input and $\xi\leftarrow\Theta(\alpha)$.}
  \STATE{ Implement by a hash table augmented with backpointers and threads for and enumerating preimages}
  \STATE{ Let $\signal^{()}$ be the filtration of $\signal$ by $h$}
  \STATE{ $I_{1,j}\leftarrow[N^{(1)}]$}
  \STATE{ // track back through levels of the filtration}
  \FOR{ $q\leftarrow 1$ to $\ell-1$}
    \STATE{ Call Algorithm~\ref{algo:weak} (Weak) on $I_{q,j}$,
      $\signal^{(q)}$, $\zeta\leftarrow 1/\ell$, noise
      $\eta\leftarrow\alpha/\ell$, sparsity $s$, and $B=2s$, getting
      $\widehat\signal$}
    \IF{ $q<\ell$}
      \STATE{
        $I_{q+1,j}\leftarrow\bigcup_{b\in\supp{\widehat\signal}}\mysplit{b}$}
    \ENDIF
    \ENDFOR
  \STATE{ $I\leftarrow \bigcup_j I_{\ell,j}$}
\ENDFOR
\STATE{ Call Algorithm~\ref{algo:weak} (Weak) on $\signal$, $I$,
  $\zeta\leftarrow 1/6$, noise $\eta\leftarrow\Omega(\alpha)$,
  sparsity $s$, $B=2s$, getting
  $\widehat\signal$}
\RETURN{ $\widehat\signal$}
\end{algorithmic}
\end{algorithm}

\begin{lemma}
\label{lem:fastweak}
With proper instantiations of constants, and with fixed values
$\zeta=1/2$, $B=2s$, and $I=[N]$,
Algorithm~\ref{algo:sublinear} is a correct Weak system
(Definition~\ref{def:weak}).  The number
of measurements is
$O(\ell^{8}\alpha^{-3} s\log(N/s))$ and the runtime is
$O(\ell^{5}\alpha^{-3}s(N/s)^{1/\ell})$, assuming a data structure that
uses preprocessing and space $O(\ell N/\alpha)$.
\end{lemma}

\begin{proof}
We maintain the following invariant for all $q$:
\begin{invariant}
\label{inv:sublinear}
We have $\signal=\mathbf{y} + \mathbf{w} + \mathbf{z}$, where
$\supp{\mathbf{y}}\subseteq \bigcup_j I_{q,j}$,
$|\supp{\mathbf{w}}|\le (s/6)(1+(q-1)/\ell)$, and
$\nerr{\mathbf{z}}_1\le 1+\alpha(q-1)/\ell$.  Elements of $\mathbf{y}$
dominate their buckets.  The size of $\bigcup_j I_{q,j}$ is
$s(N/s)^{1/\ell}$.
\end{invariant}
\iffull
The invariant holds at initialization by Lemma~\ref{lem:filt}
(Filtration).  This is
because the
elements in $\mathbf{y}$ can be assumed to be of
magnitude at least $\alpha/s$ and to dominate their buckets, while the
filtration process preserves the noise $\ell_1$ norm.  The invariant
is maintained as $q$ increases by Lemma~\ref{lem:weak} (Weak).  The
failure probability can be taken small enough so that we
can take a union bound over all $\ell\le\log(N)$ levels times the
number of choices in each level (addressed in the proof of
Lemma~\ref{lem:weak} (Weak)).

At $q=\ell$, we have $\supp{\mathbf{y}}\subseteq I$.
Each $I_{\ell,j}$ has size
$O(s)$, since it is the unsplit support of the output of
Algorithm~\ref{algo:weak}, so $|I|=O(s\ell/\alpha)$.
Also, $|\supp{\mathbf{w}}|\le s/3$ and
$\nerr{\mathbf{z}}_1\le 1+O(\alpha)$.  The
final call to Algorithm~\ref{algo:weak} (Weak)
recovers all but another $s/6$ of the support of $\mathbf{y}$, which, when combined
with $\mathbf{w}$, gives at most $s/2$ missed heavy hitters---the vector
$\widehat{\mathbf{y}}$ in the definition of a Weak system.  It also contributes
an acceptable amount $O(\alpha)$ of additional noise that, with $\mathbf{z}$,
constitute $\widehat{\mathbf{z}}$ in the definition of a Weak system.

\tightpgh{Costs.}
The number of measurements and runtime is correct by construction,
assuming the hash and split operations can be done in constant time.
This is straightforward using a hash table with appropriate pointers
for the split operation.  Such a data structure needs space $O(N)$ and
preprocessing $O(N)$ for each of the $O(\ell/\alpha)$ repetitions, for a total
of $O(\ell N/\alpha)$.  Note that the total cost, over all $\ell$ levels, is
only $O(\ell N/\alpha)$ and not $O(\ell^2 N/\alpha)$, since the contributions from the
levels form a geometric series.

In more detail, we first consider dependence on $\alpha$ and, below,
on $\ell$.  The number of measurements is proportional to
$\alpha^{-3}$, since the number of repetitions is proportional to
$\alpha^{-1}$ and the error parameter $\eta$ is proportional to
$\alpha$, so each call to Algorithm~\ref{algo:weak} requires
$\alpha^{-2}$ measurements.
The bottom ($q=1$) level takes
runtime cubic in $\alpha$, since there are $O(\ell/\alpha)$
repetitions of $I_{1,j}$ of size $O((s/\alpha)(N/s)^{1/\ell})$ and the
error parameter $\eta$ is proportional to $\alpha$.  Other levels take
runtime just $\alpha^{-2}$, since $|I_{>1,j}|$ has size
$O(s(N/s)^{1/\ell})$.

The number of measurements depends on the {\em eighth} power of
$\ell$:  one
factor for the number of repetitions in the outer loop, one factor for
the number of
levels in the inner loop, $\ell^2$ for the tighter approximation parameter
$\eta=\alpha/\ell$ and $\ell^4$ for the tighter omission
parameter $\zeta=1/\ell$, that
contribute the factor
$\eta^{-2}\zeta^{-4}$ to the costs.
The runtime of each call to Algorithm~\ref{algo:weak} is proprotional
to only the first power of $\eta\zeta^2$
times $|I|\log(N/s)$.  The
bottom level of the filtration involves a search over $I$ of size
$(s/\alpha)(N/s)^{1/\ell}$ for $\alpha\approx\epsilon/\ell$, while the
other $\ell$ levels of the filtration search over
$O(s(N/s)^{1/\ell})$.  Thus the runtime is
$O(\ell^5\alpha^{-3}s(N/s)^{1/\ell}\log(N/s))$.

Finally,
note that, $(N/s)^{1/\ell}\log(N/s)\le(N/s)^{1/(\ell-1)}$.  By
putting $\ell_0=\ell-1$, we get
\[\ell^5(N/s)^{1/\ell}\log(N/s)\le (\ell_0+1)^5(N/s)^{1/\ell_0},\]
which is
$O(\ell_0^5(N/s)^{1/\ell_0})$, so we lose the $\log(N/s)$ factor for
sufficiently large $N/s$.
\fi
\end{proof}

Some remarks follow.
Note that both the filtration and the measurement process of
Algorithm~\ref{algo:weak} involve hashing.  While the hashing of
Algorithm~\ref{algo:weak} into $B=\eta^{-1}\zeta^{-2}s$ buckets results
in $B$ measurements in each of $\eta^{-1}\zeta^{-2}\log(N/s)$ repetitions, the
hashing to create a filtration does not {\em directly} result in
measurements or any recovery-time object.  We never make
$(N/s)^{1/\ell}$ measurements---that would be too many---and we do not
instantiate the upper levels of the filtration at decode
time---instantiating a signal of length $(N/s)^{1-1/\ell}$ would take
too long.

\subsection{Toplevel System}
\label{sec:toplevel}

Finally, we give a Toplevel system.  The
construction here closely follows~\cite{GLPS10} (where it was presented for
the foreach, $\ell_2$-to-$\ell_2$ problem).
A Toplevel system is an algorithm that solves our overall problem.

\begin{definition}
\label{def:toplevel}
An
{\em approximate sparse recovery system}
(briefly, a {\bf Toplevel} system),
consists of parameters $N,k,\epsilon$, an $m$-by-$N$
{\em measurement matrix}, $\mtx{\Phi}$, and a
{\em decoding algorithm}.  Fix a vector, $\signal$, where $\signal_k$
denotes
the optimal $k$-term approximation to $\signal$.  Given the parameters
and $\mtx{\Phi}\signal$, the system
approximates $\signal$ by $\widehat \signal=\mathcal{D}(\mtx{\Phi}
\signal)$, which must satisfy
$\nerr{\widehat \signal - \signal}_1
\le (1+\epsilon)\nerr{\signal_k - \signal}_1$.
\end{definition}

\begin{thm}[Toplevel]
\label{thm:fasttoplevel}
Fix parameters $N,k,\ell$.
Algorithm~\ref{algo:fasttoplevel} (Toplevel)
returns $\widehat\signal$ satisfying
\[\nerr{\widehat\signal-\signal}_1\le(1+\epsilon)\nerr{\signal_k-\signal}_1.\]
It uses $O(\ell^{8}\epsilon^{-3}k\log(N/k))$
measurements and runs in time
$O(\ell^{5}\epsilon^{-3}k(N/k)^{1/\ell})$, using a data structure
requiring $O(\ell Nk^{0.2}/\epsilon)$ preprocessing time and storage space.
\end{thm}

\begin{algorithm}
\caption{Toplevel System}
\label{algo:fasttoplevel}
\begin{algorithmic}
\STATE{Input: $\mtx{\Phi}$, $\mtx{\Phi}\signal$, $N$, $k$, $\epsilon$}
\STATE{Output: $\widehat\signal$}
\STATE{$\widehat\signal\leftarrow0$}
\STATE{$\meas\leftarrow\mtx{\Phi}\signal$}
\FOR { $j=1$ to $\lg k$}
  \STATE {Run Algorithm~\ref{algo:sublinear} (Fast Weak) on $\meas$ with length $N$, sparsity
    $s\leftarrow k/2^j$, approx'n $\alpha\leftarrow O(\epsilon(9/10)^j)$}
  \STATE {Let $\signal'$ be the result}
  \STATE {Let $\widehat\signal=\widehat\signal + \signal'$}
  \STATE {Let $\meas=\meas-\mtx{\Phi}\signal'$}
\ENDFOR
\RETURN $\widehat\signal$
\end{algorithmic}
\end{algorithm}

\iffull
\begin{proof} [sketch]
Intuitively, the first
iteration of
Algorithm~\ref{algo:fasttoplevel} transforms a measured but unknown
$k$-sparse signal
with noise magnitude $1$ to a measured but unknown
$(k/2)$-sparse
signal with noise
$1+O(\epsilon)$.  In subsequent iterations, the sparsity $s$ decreases
(relaxes) from $k$ to $k/2$ to $k/4$ while the noise tolerance
$\alpha$ decreases (tightens) from $\epsilon$ to $(9/10)\epsilon$ to
$(9/10)^2\epsilon$, etc.  
We save a factor 2 in the
number of measurements because $s$ decreases and that more than
pays for an increase in number of measurements by the factor
$(10/9)^2$, that arises because $\eta$ decreases.  Thus measurement cost is bounded by
decreasing geometric series and so is bounded by the first term, which is
the measurement cost
of the first iteration.  Overall error is the sum of a decreasing
geometric series with ratio $9/10$, so the overall error
$\nerr{\widehat{\mathbf{z}}}_1$ remains bounded, by $1+O(\epsilon)\le 2$, with the given
algorithm.
A similar argument (with an additional wrinkle) holds
for runtime.

More formally, note that the returned vector $\widehat\signal$ has
$O(k)$ terms.  There is an invariant that
$\signal=\widehat\signal+\signaly+\signalz$, where
$\meas$ is the measurement vector for $\signaly+\signalz$,
$|\supp{\signaly}|\le k/2^j$ and
\[\nerr{\signalz}_1 \le
1+O(\epsilon)
\left[\frac9{10}+\left(\frac9{10}\right)^2+\left(\frac9{10}\right)^3+\cdots+\left(\frac9{10}\right)^j\right]\]
after $j$ iterations.
This is true at initialization, where
$\signaly=\signal_k$ and
$\nerr{\signalz}_1=\nerr{\signal-\signal_k}_1=1$.
At termination, $\signal=\widehat\signal+\signalz$, with
$\nerr{\signalz}_1\le 1+O(\epsilon)$, since the infinite geometric series
sums to 3.
Maintenance of the loop invariant follows from correctness of the Weak
algorithm.

Using the bound on measurements for Algorithm~\ref{algo:sublinear}, the
number of measurements used by Algorithm~\ref{algo:fasttoplevel}
is proportional to
\begin{eqnarray*}
\sum_j \ell^8\epsilon^{-3}(k/2^j)\log(N2^j/k)(10/9)^{2j}
& \le & \ell^8\epsilon^{-3}k\log(N/k)\sum_j (50/81+o(1))^j\\
& = & O(\ell^8\epsilon^{-3}k\log(N/k)).
\end{eqnarray*}

Similarly, using the runtime bound for Algorithm~\ref{algo:weak} and
writing $s(N/s)^{1/\ell}$ as $s^{1-1/\ell}N^{1/\ell}$, the
runtime of Algorithm~\ref{algo:fasttoplevel} is proportional to
\begin{eqnarray*}
\sum_j \ell^5\epsilon^{-3}(k/2^j)^{1-1/\ell}N^{1/\ell}(10/9)^{3j}
& \le & \ell^5\epsilon^{-3}k(N/k)^{1/\ell}
    \sum_j \left[(10/9)^3 2^{1/\ell - 1}\right]^j\\ 
& \le & \ell^5\epsilon^{-3}k(N/k)^{1/\ell}
        \sum_j \left[(10/9)^3 2^{-1/2}\right]^j,\mbox{\quad since
          $\ell\ge 2$}\\
& \le & \ell^5\epsilon^{-3}k(N/k)^{1/\ell} \sum_j 0.97^j\\
& \le & O(\ell^5\epsilon^{-3}k(N/k)^{1/\ell}).
\end{eqnarray*}

Finally, the storage space for hash tables in
Algorithm~\ref{algo:sublinear} is $N$ for each of
$\ell/\alpha$ repetitions.  This is dominated by the smallest
$\alpha$, which is
$\epsilon(9/10)^{\lg k} = \epsilon k^{-\lg 10/9}
\ge \epsilon k^{-0.2}$,
giving $N\ell k^{0.2}/\epsilon$ space and preprocessing.
For any constant
real-valued $c>0$, this can be improved to $(1/c)^{O(1)}k^c$ by
replacing $9/10$ with $1-c$ and $\epsilon$ with $c\epsilon$.  This
will also increase the runtime and number of measurements by a
constant factor that depends on $c$.
\end{proof}
\fi

\section{Open Problems}
\label{sec:open}

In this section, we present some generalizations of our algorithm that
we leave as open problems.

\tightpgh{Small space.}  Above we presented an algorithm that used
superlinear space to store and to invert a hash
function.
The
amount of space is partially excuseable because it can be amortized
over many instances of the problem, {\em i.e.}, many signals.  It also
has the advantage over a hash function that hash operations can be
performed simply in time $O(1)$.  It should be possible, however, to use a standard
hash function instead of a hash table to avoid the space requirement,
though the runtime will likely increase.  We leave
as an open problem a fuller treatment of these tradeoffs.

\tightpgh{Column Sparsity.}
An advantage in sparse recovery is the sparsity of the measurement
matrix, $\mtx{\Phi}$.  Our matrix can easily seen to have at most
$(\ell/\epsilon)^{O(1)}\log(N/k)\log(k)$ non-zeros per column, {\em i.e.}, there
is no leading
factor of $k$.  But we have not optimized $\mtx{\Phi}$ for column
sparsity and we leave that for future work. 

\tightpgh{Post-measurement noise.} Many algorithms in the literature
give, as input to the decoding algorithm, not $\mtx{\Phi}\signal$, but
$\mtx{\Phi}\signal+\noise$, where $\noise$ is an arbitrary noise vector.
The algorithm's performance must degrade gracefully with $\nerr{\noise}$
(usually the 2-norm of $\noise$).
It can be seen that our algorithm does
tolerate substantial noise, but in $\ell_1$ norm.  We leave to future
work full analysis and possible
improved algorithms.

\tightpgh{Lower overhead in number of measurements.} The approach we
present produces a Toplevel system from a Weak system, using a
filtration.  It
has a blowup factor of $\ell^8$ in the number of measurements
over a weak system,
where $\ell>1$ is an {\em integer}.  Thus the blowup factor in number
of measurements for a time-$\sqrt{kN}\log(N/k)$ algorithm is at least
$256$, even (implausibly) ignoring all overhead and other constant
factors.  This should be improved.

\tightpgh{Simplify.}
In~\cite{NT08}, the authors take a different approach to fast
algorithms.  They argue that a small number of Fourier transforms of
length $N$ in a simple algorithm that takes linear time with a DFT
oracle will be faster in practice than an algorithm that is
asymptotically sublinear.  They give an algorithm, CoSaMP, with
runtime slightly greater than $N$, under a plausible assumption about
random row-submatrices of the Fourier matrix and a bound on the
``dynamic range'' of the problem, {\em i.e.} the ratio of
$\nerr{\signal}_2$ to $\nerr{\signal-\signal_k}_2$.

In the spirit of that paper, it would be good to use our speedup
approach under the same assumptions as their paper, with $\ell=2$.
That is, ideally, we would want to double or triple the number of DFTs in the original
CoSaMP, but reduce the length of the DFTs from $N$ to approximately
$\sqrt{kN}$.  Our algorithm also suffers considerable overhead in
converting a Weak algorithm into a Toplevel algorithm---a significant
flaw if the goal is a simple, low-overhead algorithm---but CoSaMP has a
similar iterative structure and it is conceivable that our
Weak-to-Toplevel overhead can be combined subadditively with CoSaMP's
iterative overhead.

\iffull
\relax
\fi

\begin{center}
\textbf{Acknowledgement}
\end{center}
We thank Anna Gilbert, Yi Li, Hung Ngo, Atri Rudra, and Mary Wootters
for discussions and for reading an earlier draft.

\bibliographystyle{alpha}
\bibliography{l1.bib}

\end{document}